\documentclass[11pt]{article}

\usepackage[T1]{fontenc}
\usepackage[latin1]{inputenc}
\usepackage{mathtools,amsthm,amssymb}
\usepackage{mathabx}
\usepackage{mathrsfs}

\usepackage{graphicx}
\usepackage[round]{natbib}
\usepackage{amsmath,amssymb}
\usepackage[hyphens]{url}  
\usepackage{float}

\usepackage{xcolor}
\usepackage{booktabs,tabularx,ragged2e}

\usepackage{booktabs,tabularx,ragged2e}
\newcolumntype{C}{>{\Centering\arraybackslash}X} 
\newcommand{\ra}[1]{\renewcommand{\arraystretch}{#1}}

\usepackage{authblk}

\usepackage{hyperref}

\usepackage{algorithm2e}
\usepackage{algorithmic}
\makeatletter
\renewcommand{\@algocf@capt@plain}{above}
\makeatother

\usepackage{setspace}
\usepackage{xcolor}

\newcommand{\revision}[1]{#1}

\usepackage[margin=1in]{geometry}

\newtheorem{proposition}{Proposition}[section]

\newtheorem{definition}[proposition]{Definition}
\newtheorem{example}[proposition]{Example}

\newcommand{\dif}{\mathrm{d}}

\numberwithin{equation}{section}

\begin{document}

\title{Non-reversible guided Metropolis kernel}

\author[$1$]{Kengo Kamatani}
\author[$1$]{Xiaolin Song}
\affil[$1$]{{\small Osaka University}}
\date{}

\maketitle

\begin{abstract}
We construct a class of non-reversible Metropolis kernels as a multivariate extension of the guided-walk kernel proposed by \citet{Gustafson}. The main idea of our method is to introduce a projection that maps a state space to a totally ordered group. By using Haar measure, we construct a novel Markov kernel termed Haar-mixture kernel, which is of interest in \revision{its} own right. This is achieved by inducing a topological structure to the totally ordered group. Our proposed method, the $\Delta$-guided Metropolis--Haar kernel,  is constructed by using the Haar-mixture kernel as a proposal kernel. The proposed non-reversible kernel \revision{is} at least 
$10$ times better than the random-walk Metropolis kernel and Hamiltonian Monte Carlo kernel for the logistic regression and a discretely observed stochastic process in terms of effective sample size per second. 
\end{abstract}

\section{Introduction}
\subsection{Non-reversible Metropolis kernel}
\label{subsec:non-reversible}

Markov chain Monte Carlo methods have become essential tools in Bayesian computation.  Bayesian statistics \revision{has} been strongly influenced by the evolution of the methods. This influence \revision{is} well expressed in \citet{Robert_2011,Green_2015}. However, the applicability of traditional Markov chain Monte Carlo methods is limited for some statistical problems involving large data sets. This motivated researchers to work on new kinds of Monte Carlo methods, such as piecewise deterministic Monte Carlo methods \citep{MR3832232,MR3911113}, divide-and-conquer methods \citep{wang2013parallelizing,neiswanger,Scott_2016}, approximate subsampling methods \citep{Welling2011,machenfox}, and   non-reversible Markov chain Monte Carlo methods. 

In this paper, we focus on non-reversible Markov chain Monte Carlo methods. Reversibility refers to the sophisticated balancing condition (detailed-balance condition) which makes the Markov kernel invariant with respect to the probability measure of interest. Although reversible Markov kernels form a nice class  \citep{MR834478,MR1448322,MR1915532,Kontoyiannis_2011}, the condition is not necessary for the invariance.  Breaking reversibility sometimes improves the convergence properties of Markov chains \citep{edsjaa.10.2307.295968519930101,MR1789978,andrieu2019peskuntierney}. 

However, without the sophisticated balancing condition, constructing \revision{a} Markov chain Monte Carlo method is not an easy task. There are many efforts working in this direction but still there are large gap\revision{s} between the theory and practice. The guided-walk method for probability measure\revision{s} on one-dimension Euclidean space was proposed by \citet{Gustafson} which sheds light on this direction. Its multivariate extension has also \revision{been} studied in \citet{MR3905547} but \revision{is} still based on \revision{a} one dimensional Markov kernel. In this paper we consider a general multivariate extension of \citet{Gustafson}, termed guided Metropolis kernel.  To do this, we first briefly describe their method. 

\revision{
In the algorithm proposed in \citet{Gustafson}, a direction variable is attached to each state $x\in\mathbb{R}$, which is either the positive $(+)$ direction or the negative $(-)$ direction. If the positive direction is attached, the new proposed state is
\begin{equation}
x+|w|
\label{eq:sum}
\end{equation}
where $x$ is the current value and $w$ is the random noise. If the negative direction is attached, the new proposed state is
\begin{equation*}
x-|w|. 
\end{equation*}
The proposed state is accepted as the new state with the so-called acceptance probability. 
If the proposed state is accepted, the new state is assigned the same direction as the previous state. Otherwise, the opposite direction is assigned to the new state, and the new state is same as the previous state. 
}

\revision{
If we want to generalise this procedure to a more general state space, say $E$, we may need to interpret the summation operator $+$ in (\ref{eq:sum}) differently, since, for example, $\mathbb{R}_+$ is not closed with the operation. So we have to find a state space that has a suitable summation operator, in other words,  a group structure. For this reason, we consider an abstract setting throughout in this paper, as this is the most natural way to describe our setting and algorithms.  
}

More precisely, the main idea of our method is to introduce a projection which maps \revision{a} state space $E$ to  a totally ordered group. By this ordering we will decompose any Markov kernel into a sum of positive ($+$) and negative ($-$) directional sub Markov kernels. By using rejection sampling, two sub Markov kernels are normalised to be positive and negative Markov kernels. Then we can construct a non-reversible Markov kernel on  $E\times\{-,+\}$  by \revision{the} systematic-scan Gibbs sampler. Similar idea\revision{s} can be found in \citet{Gagnon2020LiftedSF} for a discrete state space case. 

Usually, total masses of sub Markov kernels are quite different which results in inefficiency of rejection sampling. To avoid this issue, we focus on the case where the total masses are the same. 
However,  it is non trivial to find such a Markov kernel. 
By using Haar measure, we introduce a novel Markov kernel termed Haar-mixture kernel, that  \revision{has this} property. This is achieved by introducing \revision{a} topological structure to the totally ordered group and $E$. 
Our proposed method, 
the $\Delta$-guided Metropolis--Haar kernel, is constructed by using the Haar-mixture kernel as a proposal kernel. By using this, we introduce many non-reversible $\Delta$-guided Metropolis--Haar kernels which are of practical interest.
\subsection{Literature review}
Here we briefly review the existing literature which has studied non-reversible Markov kernels that modify reversible Metropolis kernels. 
 First of all, products of reversible Markov kernels are not reversible in general. For example, the systematic-scan Gibbs sampler is usually non-reversible. 
 
The so-called lifting method was considered in, for example, 
 \citet{MR1789978, TURITSYN2011410, Vucelja_2016,Gagnon2020LiftedSF}. In this method, a Markov kernel is lifted to an augmented state space by splitting the Markov kernel into two sub-Markov kernels.   An \revision{auxiliary variable} chooses which kernel should be followed. The guided-walk kernel \citep{Gustafson} and the method we are proposing are classified into this category.  Another approach is preparing two Markov kernels in advance and constructing a systematic-scan Gibbs sampler  as in \citet{MR3905547}. 
 
The Hamiltonian Monte Carlo kernel has an \revision{auxiliary variable} by construction. Therefore, a systematic-scan Gibbs sampler can naturally be defined, as in \citet{HOROWITZ1991247}. Also, \citet{pmlr-v70-tripuraneni17a} constructed a different non-reversible kernel which twists the original Hamiltonian Monte Carlo kernel. See also \citet{sherlock2017discrete,ludkin2019hug}. 
 
An important exception that does not introduce an \revision{auxiliary variable} is  \citet{MR3538633} that introduces an anti-symmetric part into the acceptance probability so that the kernel becomes non-reversible while preserving $\Pi$-invariance, \revision{where a Markov kernel $P$ is called $\Pi$-invariant if $\int_{x\in E}\Pi(\dif x)P(x,A)=\Pi(A)$}. See also \citet{neal2020nonreversibly} that avoids requiring an additional \revision{auxiliary variable} by focusing on the uniform distribution that is implicitly used for the acceptance-rejection procedure in the Metropolis algorithm. 

\revision{
In this paper, non-reversible Markov kernels are designed using the Haar measure. The use of the Haar measure in the Monte Carlo context is not new. \citet{LiuWu99} used the Haar measure to improve the convergence speed of the Gibbs sampler, which was further developed by \citet{LiuSabattiB00, HobertMarchev08}. Also, the Haar measure is a popular choice of prior distribution in the Bayesian context \citep{MR1234489,RCP,MR2247439}. Markov chain Monte Carlo methods with models using the prior distribution are naturally related to the Haar measure. 
}

\subsection{Construction of the paper}

\revision{
The main objective of this paper is to present a framework for the construction of a class of non-reversible kernels, which are described in Section \ref{sec:guided}. Sections \ref{sec:haar-mixture} and \ref{sec:unbiased} are devoted to introducing some useful ideas for the construction of the non-reversible kernels. }

\revision{
Section \ref{subsec:reversibilty} contains an introduction to some reversible kernels, such as the convolution-type construction of reversible kernels and Metropolis kernels. In Section \ref{subsec:haar-mixture}, we introduce the Haar-mixture kernel and the Metropolis--Haar kernel. The Metropolis--Haar kernel is useful in its own right, although it does not have non-reversible property. Moreover, it is actually a key Markov kernel for non-reversible kernels. However, the connection to non-reversible kernels is explained in Section \ref{sec:unbiased} rather than Section \ref{sec:haar-mixture}. 
}

\revision{
In Section \ref{sec:unbiased} we introduce three properties, unbiasedness, random-walk, and sufficiency properties.} 
\revision{These properties are introduced from Section \ref{subsec:unbiasedness} to Section \ref{subsec:sufficiency} sequentially.  As described in Section \ref{subsec:non-reversible}, our construction of the non-reversible kernel is based on a Markov kernel that generates a state in the positive and negative directions with equal probability. This property is referred to as unbiasedness in Section \ref{subsec:unbiasedness} which is the sufficient condition for constructing non-reversible kernels. In Section \ref{subsec:random-walk}, we introduced a more specific form of the unbiasedness property, the random-walk property. In Section \ref{subsec:sufficiency}, we introduce the sufficiency property to describe a specific form of the random-walk property using the Haar-mixture kernel introduced in Section \ref{subsec:haar-mixture}. Section \ref{subsec:multivariate} describes how to generalise a one-dimensional unbiased kernel to a multivariate kernel. 
}

\revision{
Section \ref{sec:guided} is the section for non-reversible kernels. In Section \ref{subsec:delta-guided} we introduce a class of non-reversible kernels, the $\Delta$-guided Metropolis kernel. We focus on the $\Delta$-guided Metropolis--Haar kernel, which is a $\Delta$-guided Metropolis kernel using Haar-mixture kernel. In Section \ref{subsec:step-by-step}, we show step-by-step instructions for constructing $\Delta$-guided Metropolis--Haar kernels. Some examples can be found in Section \ref{subsec:example}.
}

\revision{
In Section \ref{sec:multiplicative}, some simulations for the $\Delta$-guided Metropolis--Haar kernel based on the autoregressive kernel are studied. Also, numerical analyses for $\Delta$-guided Metropolis--Haar kernels on $\mathbb{R}_+^d$ are studied in Section \ref{sec:additive}. Some conclusions and discussion can be found in Section \ref{sec:discussion}. 
}

\subsection{Some group related concepts}

\revision{
A set $G$ is a \textit{totally ordered set} if it has a binary relation $\le$ which satisfies three properties: 
\begin{itemize}
    \item $a\le b$ and $b\le a$ implies $a=b$,
    \item  if $a\le b$ and $b\le c$, then $a\le c$,
    \item $a\le b$ or $b\le a$ for all $a, b\in G$.
\end{itemize}
We call $\le$ an order relation. The totally ordered set $G$ can be equipped with the order topology induced by $\{g\in G:g\le a\}$ and $\{g\in G:a\le g\}$ for $a\in G$. 
A Borel $\sigma$-algebra is generated from the order topology. 
}

\revision{
A group $(G,\times)$ is an \textit{ordered group} if there is an order relation $\le$ such that 
\begin{equation}
a\le b \Longrightarrow ca\le cb\ \mathrm{and}\ ac\le bc
\label{eq:axiom-group-ordering}
\end{equation}
for $a,b,c\in G$. 
}

\revision{
A group $(G,\times)$ with a topology on $G$ is called a \textit{topological group} if its group actions $(g,h)\mapsto gh$ and $g\mapsto g^{-1}$ are continuous.  
If $G$ is locally compact and Hausdorff, it is called a \textit{locally compact topological group}. For any locally compact topological group, there is a left and right Haar measures. The group is called \textit{unimodular}  if the left Haar measure and the right Haar measure coincide up to a multiplicative constant. See \cite{MR0033869} for the detail. 
}

\revision{
The set $E$ is a \textit{left $G$-set}, if there exists a left-group action $(g,x)\mapsto gx$ from $G\times E$ to $E$ such that $(e,x)=x$ and $(g,(h,x))=(gh,x)$ where $e$ is the identity and $g, h\in G, x\in E$. We denote $gx$ for $(g,x)$. 
In this paper, any map $\Delta: E\rightarrow G$ is called a statistic when $G$ is a totally ordered set. A statistic is called a $G$-statistic if  $\Delta gx=g\Delta x$ for $g\in G$ and $x\in E$ and if $G$ is an ordered group. 
}

\section{Haar-mixture kernel}
\label{sec:haar-mixture}

\subsection{Reversibility and Metropolis kernel}
\label{subsec:reversibilty}


Before analysing the non-reversible Markov kernel, we first recall the definition of reversibility. Reversibility is important throughout the paper since our construction of a non-reversible Markov kernel is based on class\revision{es} of reversible Markov kernels. 
A Markov kernel $Q$ on a measurable space $(E,\mathcal{E})$ is $\mu$-reversible  for a $\sigma$-finite measure $\mu$ if
\begin{equation}\label{eq:reversible}
\int_A\mu(\dif x)Q(x,B)=\int_B\mu(\dif x)Q(x,A)
\end{equation}
for any $A, B\in\mathcal{E}$. If $Q$ is $\mu$-reversible, then $Q$ is $\mu$-invariant. There is a strong connection between ergodicity and $\mu$-reversibility. 
See \citet{MR834478,MR1448322,MR1915532,Kontoyiannis_2011}.

\revision{
As we mentioned above, our non-reversible Markov kernel is based on a class of reversible kernels. 
Suppose that $\mu$ is a probability measure on $(E,\mathcal{E})$ where $E$ is closed by a summation operator. A simple approach to construct a reversible kernel is to first describe $\mu$ as an image measure  of a convolution of probability measures $\mu_Y, \mu_Z$ under a measurable map $f$, i.e., $\mu=(\mu_X*\mu_Y)\circ f^{-1}$. Here, an image measure of a measure $\mu$ under a map $f:E\rightarrow E$ is defined by 
$$
\mu\circ f^{-1}(A)=\mu(\{x\in E: f(x)\in A\}), 
$$
and a convolution of $\mu_1$ and $\mu_2$ is defined by 
$$
(\mu_1*\mu_2)(A)=\int_E\mu_1(A-x)\mu_2(\dif x)
$$
where $A-x=\{y\in E: x+y\in A\}$. 
Then define independent random variables 
$Y_1, Y_2\sim\mu_Y$ and $Z\sim \mu_Z$. Finally, construct $Q$ as the conditional distribution of  $X_2=f(Y_2+Z)$ given $X_1=f(Y_1+Z)$}. Then the probabilities in (\ref{eq:reversible}) are $\mathbb{P}(X_1\in A, X_2\in B)$ and 
$\mathbb{P}(X_1\in B, X_2\in A)$ which are the same by construction.  
We refer to this as the convolution-type construction. 

Let $\mathbb{R}_+=(0,\infty)$. 
Let $I_d$ be the $d\times d$-identity matrix. 

\begin{example}[Autoregressive kernel]
\label{ex:ar-intro}
We first describe the well-known autoregressive kernel resulting from the above convolution-type construction.  
Let $\rho\in (0,1]$ and $M$ be a $d\times d$ positive definite symmetric matrix, and let $x_0\in\mathbb{R}^d$. 
\revision{
Further, let $\mathcal{N}_d(x, M)$ be the normal distribution with mean $x\in\mathbb{R}^d$ and covariance matrix $M$. 
By the reproductive property of the normal distribution, $\mu=\mathcal{N}_d(x_0, M)$ is a convolution of probability measures $\mu_Y=\mathcal{N}_d(0,\rho M)$ and $\mu_Z=\mathcal{N}_d(0,(1-\rho) M)$ with $f(x)=x_0+x$ in the notation above. 
Then the random variable $X_1$ and $X_2$ in the above notation follow $\mu$ with covariance 
$$
\operatorname{Cov}(X_1,X_2)=\operatorname{Var}(Z)=(1-\rho)M. 
$$
By the change-of-variables formula, the conditional distribution $Q(x,\cdot)=\mathbb{P}(X_2\in\cdot\ |X_1=x)$ is the autoregressive kernel, which is defined as 
\[
Q(x,\cdot)=\mathcal{N}_d(x_0+(1-\rho)^{1/2}~(x-x_0), \rho M). 
\]
Due to the nature of convolution, it is $\mu=\mathcal{N}_d(x_0, M)$-reversible. 
}
\end{example}

\begin{example}[Beta-Gamma kernel]
\label{ex:thinnedgamma-intro}
Let $\mathcal{G}(\nu,\alpha)$ be the Gamma distribution with shape parameter $\nu$ and rate parameter $\alpha$. 
Let $\mu=\mathcal{G}(k,1)$, $\mu_Y=\mathcal{G}(k(1-\rho), 1)$ and $\mu_Z=\mathcal{G}(k\rho,1)$ and $f(x)=x$ where $k\in\mathbb{R}_+$ and $\rho\in (0,1)$. The conditional distribution of $b:=Z/X_1$ given $X_1$ in the notation above, is $\mathcal{B}e(k\rho,k(1-\rho))$, where $\mathcal{B}e(\alpha,\beta)$ is the Beta distribution with shape parameters $\alpha$ and $\beta$. Therefore, the conditional distribution $Q(x,\dif y)=\mathbb{P}(X_2\in\dif y|X_1=x)$ on $E=\mathbb{R}_+$, called
Beta-Gamma (autoregressive) kernel in this paper, is given by 
\[
y = b x+ c,\ b\sim\mathcal{B}e(k\rho,k(1-\rho)),\  c\sim\mathcal{G}(k(1-\rho),1),
\]
where $b, c$ are independent, and $c$ corresponds to $Y_2$ in the above notation.  
The kernel is $\mu=\mathcal{G}(k,1)$-reversible by construction. See \citet{MR986584}. 
\end{example}

\begin{example}[Chi-squared kernel]\label{ex:chisq}
We construct a $\mu=\mathcal{G}(L/2,1/2)$-reversible kernel for
$L\in\mathbb{N}$. 
Let $\mu_Y=\mathcal{N}_L(0,\rho I_L)$ and $\mu_Z=\mathcal{N}_L(0,(1-\rho) I_L)$ and $f(x_1,\ldots, x_L)=\sum_{l=1}^Lx_l^2$. By the reproductive property, if $Y_1, Y_2\sim\mu_Y$ and $Z\sim\mu_Z$ then $X_i':=Y_i+Z\sim\mathcal{N}_L(0, I_L)$.
Therefore, $X_i=f(X_i')\sim\mu$ since $\mu$ is the Chi-squared distribution with $L$-degrees of freedom. 
The conditional distribution $Q(x,\mathrm{d}y)=\mathbb{P}(X_2\in\mathrm{d}y|X_1=x)$ is $\mu$-reversible by construction. 
We show that the conditional distribution is given by 
\begin{equation}
y = \left[\left\{(1-\rho)~x\right\}^{1/2}+\rho^{1/2}~w_1\right]^2+\sum_{l=2}^L\rho~w_l^2,
\label{eq:non-central-chi-squared}
\end{equation}
where $w_1,\ldots, w_L$ are independent and follow the standard normal distribution. 
To see this, first note that the law of $\rho^{-1/2}X_2'$ given $X_1'=x'$ is $\mathcal{N}_L(\rho^{-1/2}(1-\rho)^{1/2}x',I_L)$. Then 
the law of $\rho^{-1}X_2=f(\rho^{-1/2}X_2')$ given $X_1'=x'$ is the non-central Chi-squared distribution with $L$-degrees of freedom and the non-central parameter $f(\rho^{-1/2}(1-\rho)^{1/2}x')=\rho^{-1}(1-\rho)x$. The expression (\ref{eq:non-central-chi-squared}) follows from the property of the non-central Chi-squared distribution. 
\end{example}

The Metropolis algorithm is a clever way to construct a reversible Markov kernel with respect to a given probability measure, $\Pi$. The following definition is somewhat broader than the usual one. It even includes the independent Metropolis--Hastings kernel, which is usually classified as a Metropolis--Hastings kernel and not a Metropolis kernel. An important feature of this kernel compared to the more general Metropolis--Hastings kernel is that we do not need to know the explicit density function of the proposed Markov kernel $Q(x,\cdot)$. 

\begin{definition}[Metropolis kernel]
\label{def:mh}
Let $\mu$ be a measure, and let $\Pi$ be a probability measure with probability density function $\pi(x)$ respect to $\mu$. Let $Q$ be a $\mu$-reversible Markov kernel. 
A Markov kernel $P$ is called a Metropolis kernel of $(Q,\Pi)$ if 
\begin{equation}
\begin{split}
P(x,\dif y)&=Q(x,\dif y)\alpha(x,y)\\
&\quad+\delta_x(\dif y)\left\{1-\int_EQ(x,\dif y)\alpha(x,y)\right\}
\end{split}
\nonumber
\end{equation}
for 
\begin{equation}\label{eq:acceptance}
    \alpha(x,y)=\min\left\{1, \frac{\pi(y)}{\pi(x)}\right\}. 
\end{equation}
The function $\alpha$ is called the acceptance probability, and Markov kernel $Q$ is called the proposal kernel. 
\end{definition}

A Metropolis kernel $P$ is $\Pi$-reversible.  It is easy to create a Metropolis version of the proposal kernels presented in Examples \ref{ex:ar-intro}-\ref{ex:chisq}. 

\subsection{Haar-mixture kernel}
\label{subsec:haar-mixture}
 
We introduce Markov kernels using the Haar measure. The Haar measure enables us to construct a random walk on a locally compact topological group, which is a crucial step towards obtaining non-reversible Markov kernels in this paper. The connection between the Markov kernels and the random walk will be made clear in Section \ref{sec:unbiased}, and the connection with non-reversible Markov kenrels will be clear in Section \ref{sec:guided}. 

The idea of constructing Haar-mixture kernels is to introduce an auxiliary variable $g$ corresponding to the scaling parameter or the shift parameter of the state space. We set a prior distribution on $g$. In each iteration of the random number generation, the parameter $g$ is generated from the conditional distribution given the state space using the prior distribution. The Haar-mixture kernel uses the Haar measure for the prior distribution of $g$. 
As commented above, the reason for using the Haar measure will be made clear in later sections. 

Let $(G,\times)$ be a locally compact topological group equipped with the Borel $\sigma$-algebra. Let $E$ be a left $G$-set. We assume that $E$ is equipped with a $\sigma$-algebra $\mathcal{E}$ and the left-group action is jointly measurable. Let $Q$ be a $\mu$-reversible Markov kernel on $(E,\mathcal{E})$,
where $\mu$ is a $\sigma$-finite measure. Let
\[
Q_g(x,A)=Q(gx, gA)\ (x\in E, A\in\mathcal{E}, g\in G)
\]
where $gA=\{gx: x\in A\}\in\mathcal{E}$. 
Then $Q_g$ is $\mu_g$-reversible where 
\[
\mu_g(A)=\mu(gA). 
\]
Let $\nu$ be the right Haar measure on $G$. It satisfies $\nu(Hg)=\nu(H)$ and where $Hg=\{hg:h\in H\}\subset G$. Set
\begin{equation}\label{eq:mu_star}
\mu_*(A)=\int_{g\in G}\mu_g(A)\nu(\dif g)\ (A\in\mathcal{E}). 
\end{equation}
Assume that $\mu_*$ is a $\sigma$-finite measure. 
Then $\mu_*$ is a left-invariant measure. Indeed,
\begin{align*}
    \mu_*(aA)&=\int_{b\in G}\mu_b(aA)\nu(\dif b)\\
&=\int_{b\in G}\mu(baA)\nu(\dif b)\\
&=\int_{b\in G}\mu(bA)\nu(\dif b)\\
&=\mu_*(A). 
\end{align*}
Suppose that $\mu$ is absolutely continuous with respect to $\mu_*$. Then $(g,x)\mapsto \dif\mu_g/\dif\mu_*(x)$ is jointly measurable. This is because
$\dif\mu_g/\dif\mu_*(x)=\dif\mu/\dif\mu_*(gx)$ by the left-invariance of $\mu_*$, and $(g,x)\mapsto gx$ is assumed to be jointly measurable. 
Let
\begin{equation}
K(x,\dif g)=\left.\frac{\dif\mu_{g}}{\dif\mu_*}(x)\nu(\dif g)~\right. 
\label{eq:markov-kernel-k}
\end{equation}
By the Radon--Nikod\'ym theorem, $K(x,G)=1$ $\mu_*$-almost surely. 
Define
\begin{equation}\label{eq:haar_mixture}
Q_*(x, A)=\int_{g\in G}K(x,\dif g)Q_g(x, A). 
\end{equation}

\begin{definition}[Haar-mixture kernel]
The Markov kernel $Q_*$ defined by (\ref{eq:haar_mixture}) is called the Haar-mixture kernel of $Q$. 
\end{definition}

\begin{example}[Autoregressive mixture kernel]
\label{ex:ar-mid}
Consider the autoregressive kernel in Example \ref{ex:ar-intro}. 
Let $E=\mathbb{R}^d$ and $G=(\mathbb{R}_+,\times)$, and set $(g,x)\mapsto x_0+g^{1/2}(x-x_0)$. Then the Haar measure is $\nu(\dif g)\propto g^{-1}\dif g$. A simple calculation yields $\mu_g=\mathcal{N}_d(x_0, g^{-1}M)$ and  
$Q_g(x,\cdot)=\mathcal{N}_d(x_0+(1-\rho)^{1/2}~(x-x_0), g^{-1}\rho M)$. Also, 
$\mu_*(\dif x)\propto (\Delta x)^{-d/2}\dif x$ and
$K(x,\dif g)=\mathcal{G}(d/2, \Delta x/2)$
where $\Delta x =(x-x_0)^{\top}M^{-1}(x-x_0)$. 
We have a closed form (up to a constant) of expression of $Q_*(x,\cdot)$ as follows: 
\begin{align*}
    Q_*(x,\mathrm{d}y)\propto \left[1+\frac{\Delta(y-(1-\rho)^{1/2}(x-x_0))}{\rho\Delta x}\right]^{-d}\dif x. 
\end{align*}
\end{example}

\begin{example}[Beta-Gamma mixture kernel]
\label{ex:thinnedgamma}
For the Beta-Gamma kernel in Example \ref{ex:thinnedgamma-intro}, we introduce
 \revision{the} operation $(g,x)\mapsto gx$ with  $G=(\mathbb{R}_+,\times)$. By this operation, $E=\mathbb{R}_+$ is a left $G$-set. 
We have $\mu_g=\mathcal{G}(k,g)$, and the Markov kernel $Q_g$ is the same as $Q$ replacing
$c\sim\mathcal{G}(k(1-\rho),1)$ by 
$c\sim\mathcal{G}(k(1-\rho),g)$.  
The Haar measure on $G$ is $\nu(\dif g)\propto g^{-1}\dif g$,  and hence  $\mu_*(\dif x)\propto x^{-1}\dif x$ and  $K(x,\dif g)=\mathcal{G}(k, x)$. 
\end{example}

\begin{example}[Chi-squared mixture kernel]
\label{ex:chisq-mix}
For the Chi-squared kernel in Example \ref{ex:chisq}, let $E=\mathbb{R}_+$, $G=(\mathbb{R}_+,\times)$ and set  $(g,x)\mapsto gx$. 
We have $\mu_g=\mathcal{G}(L/2, g/2)$, and the Markov kernel $Q_g$ is the same as $Q$ replacing the standard normal distribution by $\mathcal{N}(0,g^{-1})$.  
The Haar measure is $\nu(\dif g)\propto g^{-1}\dif g$. 
In this case, $K(x,\dif g)=\mathcal{G}(L/2,  x/2)$, and $\mu_*(\dif x)=x^{-1}\dif x$.
\end{example}

\begin{proposition}
The Haar-mixture kernel $Q_*$ is $\mu_*$-reversible. 
\end{proposition}

\begin{proof}
Let $A, B\in\mathcal{E}$. Since $Q_g$ is $\mu_g$-reversible, 
    \begin{align*}
        \int_{A}\mu_*(\dif x)Q_*(x, B)&=
        \int_{g\in G}\int_{x\in A}\mu_*(\dif x)K(x,\dif g)Q_g(x, B)\\
        &=
        \int_{g\in G}\int_{x\in A}\mu_g(\dif x)Q_g(x, B)\nu(\dif g)\\
        &=
        \int_{g\in G}\int_{x\in B}\mu_g(\dif x)Q_g(x, A)\nu(\dif g)\\
        &=\int_{B}\mu_*(\dif x)Q_*(x, A). 
    \end{align*}
\end{proof}

From this, we can define the following Metropolis kernel. 

\begin{definition}[Metropolis--Haar kernel]
A Metropolis kernel $P_*$ of $(Q_*,\Pi)$  is called a \textit{Metropolis--Haar kernel}
 if $Q_*$ is a Haar-mixture kernel. 
\end{definition}

The Metropolis--Haar kernel is implemented as the following algorithm, where  $\pi(x)=(\dif\Pi/\dif\mu_*)(x)$. In the algorithm, $\mathcal{U}[0,1]$ is the uniform distribution on $[0,1]$.


 \begin{algorithm}
 \caption{Metropolis--Haar kernel}
\bigskip
 \begin{algorithmic}[1]
 \setstretch{1.1}
 \renewcommand{\algorithmicrequire}{\textbf{Input:}}
 \renewcommand{\algorithmicensure}{\textbf{Output:}}
 \REQUIRE $x\in E$\\
  \STATE Simulate $g\sim K(x,\dif g)$\\
   \STATE Simulate $y\sim Q_g(x,\dif y)$\\
   \STATE Simulate $u\sim\mathcal{U}[0,1]$\\
   \STATE  If $u\le \min\{1, \pi(y)/\pi(x)\}$, set $x\leftarrow y$\\
 \RETURN $x$
 \ENSURE  $x$
 \end{algorithmic}
 \end{algorithm}
 
 The Metropolis--Haar kernel is reversible, but important in its own right. The underlying reference measure $\mu_*$ is heavier than $\mu_g$, which is expected to lead to  a robust algorithm. Examples of Metropolis--Haar kernels will be described in Section \ref{subsec:example}. 

\section{Unbiasedness, the random-walk property and sufficiency}\label{sec:unbiased}

\subsection{Unbiasedness}
\label{subsec:unbiasedness}

In this section, we introduce the unbiasedness property for efficient construction of the non-reversible kernel. 
Any measurable map $\Delta:E\rightarrow G$ is called a statistic in this paper, where $G=(G,\le)$ is a totally ordered set. 
In Section \ref{sec:guided}, a statistic $\Delta$ will guide a Markov kernel $Q(x,\dif y)$ according to the \revision{auxiliary} directional variable $i\in\{-,+\}$ as in \citet{Gustafson}. 
When the positive direction $i=+$ is selected, then $y$ is sampled according to $Q(x,\dif y)$ unless  $\Delta x\le\Delta y$ by rejection sampling. If the negative direction $i=-$ is selected, $y$ is sampled unless $\Delta y\le \Delta x$. 
It is typical that one of the rejection sampling directions has high rejection probability (see Example \ref{ex:rw}). 
To avoid this inefficiency, we consider a class of Markov kernels $Q$ such that the probabilities of the events $\Delta x\le\Delta y$ and  $\Delta y\le \Delta x$ measured by $Q(x,\cdot)$ are the same. We say $Q$ is unbiased if this property is satisfied. If the unbiasedness is violated, the rejection sampling can be inefficient because it takes a long time to exit the while loop of the rejection sampling. Therefore, the unbiasness property is necessary for efficient construction of the nonreversible kernel in our approach. 

\begin{definition}[$\Delta$-unbiasedness]
\label{def:unbiasedness}
Let $\Delta:E\rightarrow G$ be a statistic. 
We say a Markov kernel $Q$ on $E$ is $\Delta$-unbiased if
\begin{align*}
Q(x,\{y\in E:\Delta x\le\Delta y\})=
Q(x,\{y\in E:\Delta y\le\Delta x\})
\end{align*}
for any $x\in E$. 
Also, we say that two statistics $\Delta$
and $\Delta'$ from $E$ to possibly different totally ordered sets are equivalent if 
\begin{equation}
\begin{split}
Q(x,\{y\in E:\Delta x\le \Delta y\}\ominus
\{y\in E:\Delta' x\le \Delta' y\})&=0,\\ 
Q(x,\{y\in E:\Delta y\le \Delta x\}\ominus
\{y\in E:\Delta' y\le \Delta' x\})&=0
\end{split}
\nonumber
\end{equation}
for $x\in E$, where 
$A\ominus B=(A\cap B^c)\cup (A^c\cap B)$. 
\end{definition}

If $\Delta$ and $\Delta'$ are equivalent, then $\Delta$-unbiasedness implies $\Delta'$-unbiasedness.

\begin{example}[Random-walk kernel]
\label{ex:rw}
Let $v^{\top}$ be the transpose of $v\in\mathbb{R}^d$ and $\Gamma$ be a probability measure on $\mathbb{R}^d$ which is symmetric about the origin, that is, $\Gamma(A)=\Gamma(-A)$ for $-A=\{x\in E: -x\in A\}$.   
Let
$Q(x,A)=\Gamma(A-x)$. Then $Q$ is $\Delta$-unbiased for  $\Delta x=v^{\top}x$ for some $v\in\mathbb{R}^d$ since 
\begin{equation*}
\begin{split}
Q(x,\{y:\Delta x\le \Delta y\})
& =\Gamma(\{z: 0\le v^{\top}z\})\\
& =\Gamma(\{z: v^{\top}z\le 0\}). 
\end{split}
\end{equation*}

On the other hand, $Q$ is not $\Delta'$-unbiased for 
$\Delta' x=x_1^2+\cdots +x_d^2$, where $x=(x_1,\ldots, x_d)$, if $\Gamma$ is not the Dirac measure  centred on $(0,\ldots,0)$. 
In particular, if $\Gamma(\{(0,\ldots, 0)\})=0$, then 
$Q(x, \{\Delta' y\le \Delta' x\})=0$ for $x=(0,\ldots, 0)$. 
\end{example}

\subsection{Random-walk property}
\label{subsec:random-walk}

Constructing a $\Delta$-unbiased Markov kernel is a crucial step for our approach. 
However, determining how to construct a $\Delta$-unbiased Markov kernel is nontrivial. The random-walk property is the key for this construction. 

Let $G$ be a topological group. 

\begin{definition}[$(\Delta,\Gamma)$-random-walk]
\revision{A} Markov kernel $Q(x,\dif y)$ has the $(\Delta,\Gamma)$-random-walk property if 
there is a function $\Delta:E\rightarrow G$ with a  probability measure $\Gamma$ on a topological group $G$ such that 
$\Gamma(H)=\Gamma(H^{-1})$ for any Borel set $H$ of $G$ and
\begin{equation}
 Q(x, \{y\in E: \Delta y\in H\})=\Gamma((\Delta x)^{-1}H). 
\label{eq:random-walk-property}
\end{equation}
\revision{Here,  $H^{-1}=\{g\in G:  g^{-1}\in H\}$. }
\end{definition}

A typical example of a Markov kernel with the $(\Delta,\Gamma)$-random-walk property is Example \ref{ex:rw}. \revision{We assume that $(G,\le)$ is an ordered group. 
}

\begin{proposition}
\label{prop:random_to_unbiasedness}
If $Q$ has the $(\Delta,\Gamma)$-random-walk property, then $Q$ is $\Delta$-unbiased. 
\end{proposition}

\begin{proof}
Let $H=[\Delta x,+\infty)=\{g\in H: \Delta x\le g\}$. Then for the unit element $e$,  
\begin{align*}
Q(x,\{y\in E:\Delta x\le\Delta y\})
&=
Q(x,\{y\in E: \Delta y\in H\})\\ 
&=
\Gamma((\Delta x)^{-1}H)=\Gamma([e,+\infty)).
\end{align*}
Similarly, 
$Q(x,\{y\in E:\Delta y\le\Delta x\})=\Gamma((-\infty, e])$. 
Since $[e,+\infty)^{-1}=(-\infty,e]$,  $Q$ is $\Delta$-unbiased. 
\end{proof}

\subsection{Sufficiency}\label{subsec:sufficiency}

\revision{
So far in this section, we have introduced the $\Delta$-unbiasedness, which is the important property for the $\Delta$-guided Metropolis kernel in Section \ref{subsec:unbiasedness}. In Section \ref{subsec:random-walk}, we showed that the $(\Delta,\Gamma)$-random-walk property is sufficient for the $\Delta$-unbiasedness. In this subsection we will show that for the Haar-mixture kernel, the sufficiency property introduced below is sufficient for the $(\Delta,\Gamma)$-random-walk property, and for the $\Delta$-unbiasedness property. 
}

We would like to mention the intuition behind the sufficiency property. In general, the conditional law of $\Delta y$ given $x$ is not completely determined by $\Delta x$. If it is completely determined by $\Delta x$, we call $\Delta$ sufficient.  If $\Delta$ is sufficient, the equation (\ref{eq:random-walk-property}) is satisfied, although $\Gamma$ is not symmetric in general. 
When $Q$ is the Haar mixture kernel with some additional technical conditions, we will show that $\Gamma$ is symmetric thanks to the Haar measure property. 

Let $(G,\times)$ be a unimodular locally compact topological group. 
Also, let $(G,\le)$ be an ordered group, 
$E$ be a left $G$-set. In this paper, 
a statistics $\Delta:E\rightarrow G$ is called a $G$-statistics if  $\Delta gx=g\Delta x$ for $g\in G$ and $x\in E$. 
For a $\sigma$-finite measure $\Pi$ on $E$ and a $G$-statistic $\Delta:E\rightarrow G$, let $\widehat{\Pi}=\Pi\circ\Delta^{-1}$, that is, the image measure of $\Pi$ under $\Delta$. 
Let $\widehat{\mu}_*$ be the image measure of $\mu_*$ under $\Delta$. Then it is a left Haar measure, since
\begin{align*}
\widehat{\mu}_*(gH)&=\mu_*(\{y\in E: \Delta y\in gH\})\\
&=\mu_*(\{y\in E: \Delta( g^{-1}y)\in H\})\\
&=\mu_*(\{y\in E: \Delta y\in H\})\\
&=\widehat{\mu}_*(H)
\end{align*}
by the left-invariance of $\mu_*$. 
Since $G$ is unimodular, the left Haar measure $\widehat{\mu}_*$ and right Haar measure $\nu$ coincide up to a multiplicative  constant. From this fact, we can assume 
\[
\widehat{\mu}_*=\nu
\]
without loss of generality. 
Let $Q$ be a $\mu$-reversible kernel.

\begin{definition}[Sufficiency]
\label{def:sufficiency}
Let $\mu$ be a $\sigma$-finite measure. 
We call a $G$-statistic $\Delta$ sufficient if there is a Markov kernel $\widehat{Q}$ and a measurable function $h_1$ on $G$ such that 
$$
Q(x, \{y\in E: \Delta y\in H\})=\widehat{Q}(\Delta x, H)
$$
and 
$$
\frac{\dif \mu}{\dif\mu_*}(x)=h_1(\Delta x)
$$
$\mu_*$-almost surely. 
\end{definition}






By the left-invariance of $\mu_*$, we have
\begin{equation}
\label{eq:left_invariance_derivative}
    \frac{\dif\mu_g}{\dif\mu_*}(x)=h_1(g\Delta x)
\end{equation}
since
\begin{align*}
    \mu_g(A)&=\mu(gA)\\
    &=\int_{gA} h_1(\Delta x)\mu_*(\dif x)\\
    &=\int_{A} h_1(g\Delta x)\mu_*(\dif x). 
\end{align*}
Let $\widehat{\mu}$ be the image measure of $\mu$ under $\Delta$. Then $\widehat{Q}$ is $\widehat{\mu}$-reversible and 
$$
    \frac{\dif\widehat{\mu}}{\dif\nu}(a)=\frac{\dif\widehat{\mu}}{\dif\widehat{\mu}_*}(a)=h_1(a). 
$$

\begin{example}[Sufficiency of the Autoregressive mixture kernel]
\label{ex:ar}
Consider the Autoregressive  kernel $Q$ in Example \ref{ex:ar-intro} and the statistics $\Delta$ defined in Example \ref{ex:ar-mid}.
We show that $\Delta x$ is sufficient for $Q$. 
The Markov kernel $Q(x,\dif y)$ corresponds to the update 
$$
y \leftarrow x_0+(1-\rho)^{1/2}(x-x_0)+\rho^{1/2}~M^{1/2}~w
$$
where $w\sim\mathcal{N}_d(0, I_d)$. 
For $\xi=(1-\rho)^{1/2}\rho^{-1/2}M^{-1/2}(x-x_0)$, 
\[
\Delta y=\rho\left\|\xi+w\right\|^2,
\]
where $\|\cdot\|$ is the Euclidean norm. 
Therefore, $\rho^{-1}\Delta y$ conditioned on $x$ follows the non-central Chi-squared distribution with $d$ degrees of freedom and non-central parameter
$\|\xi\|^2=(1-\rho)\rho^{-1}\Delta x$. 
Hence, the law of $\Delta y$  depends on $x$ only through $\Delta x$ and hence there exists a Markov kernel $\widehat{Q}$ as in Definition \ref{def:sufficiency}. Also, a simple calculation yields $h_1(g)\propto g^{d/2}\exp(-g/2)$.  Therefore, $\Delta$ is sufficient for $Q$. 
\end{example}

\begin{example}[Sufficiency of the Beta-Gamma and Chi-squared kernels]
If $G=E$ and $\Delta x=x$ is a $G$-statistic, then it is sufficient if $\mu$ is absolutely continuous with respect to $\mu_*$.   In particular, 
for the Beta-Gamma kernel in Example \ref{ex:thinnedgamma-intro}  and Chi-squared kernel \ref{ex:chisq}, $\Delta x=x$ is sufficient. 
\end{example}



\revision{
For a measure $\nu$, we write 
$\nu^{\otimes k}$ for the $k$th product of $\nu$ defined by 
$$
\nu^{\otimes k}(\dif x_1\cdots\dif x_k)=\nu(\dif x_1)\cdots\nu(\dif x_k)
$$
for $k\in\mathbb{N}$. }

\begin{proposition}
\label{prop:unbiasedness}
Suppose a $G$-statistic $\Delta$ is sufficient for a $\mu$-reversible kernel $Q$. 
Also, suppose a probability measure 
$\widehat{\mu}(\dif a)\widehat{Q}(a,\dif b)$ on $G\times G$ is absolutely continuous with respect to $\nu^{\otimes 2}$. 
 Then $Q_*$ has the $(\Delta,\Gamma)$-random-walk property for a probability measure $\Gamma$. 
 In particular, it is $\Delta$-unbiased. 
\end{proposition}

\begin{proof}
Let $h(a,b)$ be \revision{the} Radon--Nikod\'ym derivative:
\begin{equation}
h(a,b)\nu(\dif a)\nu(\dif b)=\widehat{\mu}(\dif a)\widehat{Q}(a,\dif b). 
\nonumber
\end{equation}
By the $\widehat{\mu}$-reversibility of $\widehat{Q}$, $h(a,b)=h(b,a)$ almost surely. 
From the sufficiency property, we can rewrite $h_1$ and $\widehat{Q}$ by $h(a,b)$ and $\nu$: 
\begin{align*}
&h_1(a)=\int_{b\in G}h(a,b)\nu(\dif b),\\ &h_1(a)\widehat{Q}(a,\dif b)=h(a,b)\nu(\dif b), 
\end{align*}
$\nu$-almost. 
Together with (\ref{eq:left_invariance_derivative}), we have
\begin{align*}
&Q_*(x,\{y: \Delta y\in H\})\\
=&
  \int_{a\in G}K(x,\dif a)Q(ax, \{y: \Delta y\in aH\})\\
=&
  \int_{a\in G}\frac{\dif\mu_a}{\dif\mu_*}(x)\nu(\dif a)\widehat{Q}(a \Delta x, aH)\\
=&
  \int_{a\in G}h_1(a\Delta x)\widehat{Q}(a \Delta x, aH)\nu(\dif a)\\
=&
  \int_{a\in G}\int_{b\in H}h(a\Delta x, ab)\nu(\dif a)\nu(\dif b)\\
 =&
  \int_{a\in G}\int_{b\in H}h(a, a(\Delta x)^{-1}b)\nu(\dif a)\nu(\dif b)
\end{align*}
where the last equality follows from the right-invariance of $\nu$. 
Let
\[
 \widehat{h}(b)= \int_{a\in G}h(a,ab)\nu(\dif a). 
\]
From $h(a,b)=h(b,a)$, 
\begin{align*}
\widehat{h}(b^{-1})&= \int_{a\in G}h(a,ab^{-1})\nu(\dif a)\\
&= \int_{a\in G}h(ab,a)\nu(\dif a)=\widehat{h}(b). 
\end{align*}
By using $\widehat{h}$, we can write 
\begin{align*}
Q_*(x,\{y: \Delta y\in H\})
&=
  \int_{b\in H}\widehat{h}((\Delta x)^{-1}b)\nu(\dif b)\\
 & =
  \int_{b\in (\Delta x)^{-1}H}\widehat{h}(b)\nu(\dif b). 
\end{align*}
The above is guaranteed to have the $(\Delta,\Gamma)$-random-walk property by introducing $\Gamma(H)=\int_{a\in H}\widehat{h}(a)\nu(\dif a)$ because
$$
Q(x,\{y:\Delta y\in H\})=\Gamma((\Delta x)^{-1}H)
$$ 
and 
\begin{align*}
\Gamma(H^{-1})&=\int_{a^{-1}\in H}\widehat{h}(a)\nu(\dif a)\\
&=\int_{a\in H}\widehat{h}(a)\nu(\dif a)=\Gamma(H). 
\end{align*}
Hence, it is $\Delta$-unbiased by Proposition \ref{prop:random_to_unbiasedness}. 
\end{proof}

\subsection{Multivariate version of one-dimensional kernels} 
\label{subsec:multivariate}

Essentially, we have introduced three Markov kernels, the Autoregressive kernel, the Chi-squared kernel, and the Beta-Gamma kernel. The state space of the first kernel is a general Euclidean space and that of the last two kernels is a subspace of the one-dimensional Euclidean space. In this subsection, we consider the multivariate version of the latter two kernels. 

We present different strategies for the two kernels. For the Chi-squared kernel, there is a sophisticated structure that allows multivariate version of the state space. For the Beta-Gamma kernel, there does not seem to have a special structure and so we apply a general approach which does not require any structure. First we show how to construct a multivariate extension for the Chi-squared kernel. 

\begin{example}[Multivariate Chi-squared mixture kernel]
\label{ex:multi-chi-squared}
For the Chi-squared kernel (Examples \ref{ex:chisq}, \ref{ex:chisq-mix}), we use the operation $(g,x)\mapsto (gx_1,\ldots, gx_d)$ with $G=\mathbb{R}_+$ and $E=\mathbb{R}_+^d$. Let $Q$ be the Markov kernel defined in Example \ref{ex:chisq}. Let 
$$
\mathcal{Q}(x,\dif y)=Q(x_1,\dif y_1)\cdots Q(x_d,\dif y_d)
$$
and $\mu(\dif x)=\mathcal{G}(L/2,1/2)^{\otimes d}$. Let $\Delta x=x_1+\cdots +x_d$. 
In this case, 
$\nu(\dif g)\propto g^{-1}\dif g$ and  $\mu_g(\dif x)=\mathcal{G}(L/2,g/2)^{\otimes d}$, and $\mathcal{Q}_g$ on $\mathbb{R}^d_+$ is the product of $Q_g$ on $\mathbb{R}_+$ defined in Example \ref{ex:chisq-mix}, that is, 
$$
\mathcal{Q}_g(x,\dif y)=Q_g(x_1,\dif y_1)\cdots Q_g(x_d,\dif y_d). 
$$
Then 
$$
\mu_*(\dif x)\propto (x_1\cdots x_d)^{L/2-1}(\Delta x)^{-dL/2}\dif x_1\cdots\dif x_d
$$
and $K(x,\dif g)=\mathcal{G}(Ld/2, \Delta x/2)$. From this expression, 
$h_1(g)\propto g^{dL/2}\exp(-g/2)$. Moreover, by the property of the non-central Chi-squared distribution, the law of $\rho^{-1}\Delta y$ where $y\sim \mathcal{Q}(x,\dif y)$ is the non-central Chi-squared distribution with $dL$-degrees of freedom with the non-central parameter $(1-\rho)\rho^{-1}\Delta x$. Therefore there exists a Markov kernel $\widehat{\mathcal{Q}}(g,\cdot)$ which is the scaled non-central Chi-squared distribution for each $g$. Obviously, it has a density function with respect to $\nu$. The statistic $\Delta$ is sufficient and the multivariate version of Chi-squared mixture kernel $\mathcal{Q}_*$ is $\Delta$-unbiased from this fact. 

\end{example}

\begin{example}[Multivariate Beta-Gamma mixture kernel]
\label{ex:product}
\revision{
For the Beta-Gamma kernel (Examples \ref{ex:thinnedgamma-intro}, \ref{ex:thinnedgamma}), we use the operation $(g,x)\mapsto (g_1x_1,\ldots, g_dx_d)$ with $G=(\mathbb{R}_+^d,\times)$ and $E=\mathbb{R}_+^d$ where 
$g=(g_1,\ldots, g_d)$ and  $x=(x_1,\ldots, x_d)$. 
We define the binary operation of $G$ by 
$(x,y)\mapsto (x_1y_1,\cdots, x_dy_d)$ and the identity element by $e=(1,\ldots, 1)$. In this case, the Markov kernel $\mathcal{Q}_g$ on $\mathbb{R}^d$ is the product of $Q_g$ on $\mathbb{R}$ defined in Example \ref{ex:thinnedgamma}, that is, 
$$
\mathcal{Q}_g(x,\dif y)=Q_{g_1}(x_1,\dif y_1)\cdots Q_{g_d}(x_d,\dif y_d). 
$$
Also, we have $K(x,\dif g)=\mathcal{G}(k, x_1)\cdots\mathcal{G}(k, x_d)$ and  $\mu_*(\dif x)=(x_1\cdots x_d)^{-1}\dif x_1\cdots \dif x_d $. 
The $G$-statistic $\Delta x=x$ is sufficient, and hence the Multivariate version of Beta-Gamma mixture kernel $\mathcal{Q}_*$ is $\Delta$-unbiased by  Proposition \ref{prop:unbiasedness}. 
}
\end{example}

For $G=\mathbb{R}_+^d$ in Example \ref{ex:product}, several types of order relations are possible. Any ordering will do as long as (\ref{eq:axiom-group-ordering}) is satisfied. The popular lexicographic order depends on how we index the coordinates.  To avoid this unfavourable property, we consider the modified lexicographic order defined below. 

\begin{example}[Modified lexicographical order]
Let $G=(\mathbb{R}_+^d,\times)$.  
For $x=(x_1,\ldots, x_d)\in G$, let 
$$
s(x)_i=x_i\times\cdots\times x_d
$$ be a partial product of the vector $x$ from the $i$th element to the $d$th element. 
A version of lexicographical order $\le $ can be defined as follows.  Counting from $i=1,\ldots, d$, 
\begin{itemize}
    \item if $s(x)_i=s(y)_i$ for all $i$ or
    \item if the first index $i$ such that 
$s(x)_i\neq s(y)_i$ satisfies
$s(x)_i<s(y)_i$,
\end{itemize} 
then we write $x\le y$. 
It is not difficult to check that this ordering satisfies (\ref{eq:axiom-group-ordering}).
\end{example}

Since (\ref{eq:axiom-group-ordering}) is satisfied, the multivariate Beta-Gamma mixture kernel is $\Delta$-unbiased with this order for $G$. 
Note that the modified lexicographic order also has the same problem as that of the (un-modified) lexicographic order, that is, it depends how we index the coordinate. However, the problem occurs with probability $0$. This is because the first step of the sort (i.e. $s(x)_1<s(y)_1$ or $s(y)_1<s(x)_1$) does not depend on the order of the indexes and the first step determines the order with probability $1$. More precisely,  
$$
\Delta(x_1,\ldots, x_d)=(x_1,\ldots, x_d)
$$
with the modified lexicographical ordering and 
$$
\Delta'(x_1,\ldots,x_d)=x_1\times\cdots\times x_d
$$
in $\mathbb{R}_+$ with usual ordering are equivalent in the sense of Definition \ref{def:unbiasedness} because $\Delta'(x)=s(x)_1$. \revision{In particular, the multivariate Beta-Gamma mixture kernel is $\Delta'$-unbiased since the kernel is $\Delta$-unbiased.}  Note that $\Delta'$ is not a $G$-statistic, since it does not satisfy $\Delta'gx=g\Delta'x$. 


 \section{Guided Metropolis kernel}
\label{sec:guided}

\subsection{$\Delta$-Guided Metropolis kernel}
\label{subsec:delta-guided}

\begin{definition}[$\Delta$-guided Metropolis kernel]\label{def:guided}
For $\Delta$-unbiased Markov kernel $Q$, probability measure $\Pi$ and a measurable function $\alpha:E\times E\rightarrow [0,1]$ defined in (\ref{eq:acceptance}), we say a Markov kernel $P_G$ on $E\times \{-,+\}$ is the $\Delta$-guided Metropolis kernel of $(Q,\Pi)$ if 
\begin{align*}
P_G(x,+,\dif y,+)&=Q_+(x,\dif y)\alpha(x,y)\\
P_G(x,+,\dif y,-)&=\delta_x(\dif y)\left\{1-\int_EQ_+(x,\dif y)\alpha(x,y)\right\}\\
P_G(x,-,\dif y,-)&=Q_-(x,\dif y)\alpha(x,y)\\
P_G(x,-,\dif y,+)&=\delta_x(\dif y)\left\{1-\int_EQ_-(x,\dif y)\alpha(x,y)\right\},
\end{align*}
where
\begin{align*}
Q_+(x,\dif y)&=2Q(x,\dif y)1_{\{\Delta x <\Delta y\}}+Q(x,\dif y)1_{\{\Delta x=\Delta y\}},\\ 
Q_-(x,\dif y)&=2Q(x,\dif y)1_{\{\Delta y<\Delta x\}}+Q(x,\dif y)1_{\{\Delta x=\Delta y\}}. 
\end{align*}
\end{definition}
The Markov kernel $P_G$ satisfies the so-called $\Pi_G$-skew-reversible property
\begin{equation}
\begin{split}
\Pi_G(\dif x,+)P_G(x,+,\dif y,+)&=
\Pi_G(\dif y,-)P_G(y,-,\dif x,-),\\
\Pi_G(\dif x,+)P_G(x,+,\dif y,-)&=
\Pi_G(\dif y,-)P_G(y,-,\dif x,+),
\end{split}
\nonumber
\end{equation}
where 
\[
\Pi_G=\Pi\otimes(\delta_{-}+\delta_{+})/2. 
\]
\revision{
Here, for a probability measures $\nu$ and $\mu$, 
$(\nu\otimes\mu)(\dif x\dif y)=\nu(\dif x)\mu(\dif y)$. }
With this property, it is straightforward to check that $P_G$ is $\Pi_G$-invariant.

\begin{example}[Guided-walk kernel]
The $\Delta$-guided Metropolis kernel corresponding to the random-walk kernel $Q$  on $\mathbb{R}$ is called the guided-walk in  \citet{Gustafson}. For a multivariate target distribution,  $\Delta x=v^{\top}x$ for some $v\in\mathbb{R}^d$ was considered in \citet{Gustafson,MR3905547}. 
\end{example}

As described in Proposition \ref{prop:unbiasedness}, a Haar mixture kernel $Q_*$ is $\Delta$-unbiased if $\Delta$ is sufficient and some other technical conditions are satisfied. 
Therefore, we can construct a $\Delta$-guided Metropolis kernel $(Q_*,\Pi)$ using the Haar mixture kernel $Q_*$. 

\begin{definition}[$\Delta$-guided Metropolis--Haar kernel]
If a Haar-mixture kernel $Q_*$ is $\Delta$-unbiased, the $\Delta$-guided Metropolis kernel of $(Q_*,\Pi)$ is called the \textit{$\Delta$-guided Metropolis--Haar kernel}.
\end{definition}

The $\Delta$-guided Metropolis--Haar kernel  is given as Algorithm \ref{alg:delta}, where we let $\pi(x)=\dif\Pi/\dif\mu_*(x)$. 
 This Metropolis--Haar kernels is further discussed in detail in
 Sections \ref{subsec:step-by-step} and \ref{subsec:example}. 
 
\begin{algorithm}
 \caption{$\Delta$-guided Metropolis--Haar kernel }
\label{alg:delta} 
\smallskip
 \begin{algorithmic}[1]
 \setstretch{1.1}
 \renewcommand{\algorithmicrequire}{\textbf{Input:}}
 \renewcommand{\algorithmicensure}{\textbf{Output:}}
 \REQUIRE Input $(x,z)\in E\times \{-,+\}$\\
  \STATE Set $y=x$\\
   \STATE While $(\Delta y-\Delta x)\times z\le 0$\\
       \qquad Simulate $g\sim K(x,\dif g)$\\
       \qquad Simulate $y\sim Q_g(x,\dif y)$\\
   \STATE Simulate $u\sim\mathcal{U}[0,1]$\\
   \STATE  If $u\le \min\{1, \pi(y)/\pi(x)\}$, set $x\leftarrow y$\\
   Else set $z\leftarrow -z$ 
 \ENSURE  $(x,z)$
 \end{algorithmic}
 \end{algorithm}

Let $P$ be the Metropolis kernel of $(Q,\Pi)$. 
We now see that $P_G$ is always expected to be better than $P$ in the sense of the asymptotic variance corresponding to the central limit theorem. The inner product 
$\langle f,g\rangle=\int f(x)g(x)\Pi(\dif x)$ and the norm $\|f\|=(\langle f,f\rangle)^{1/2}$ can be defined on the space of $\Pi$-square integrable functions. Let 
$(X_0, X_1,\ldots)$ be a Markov chain with Markov kernel $P$ and $X_0\sim \Pi$. Then we define the asymptotic variance
$$
\operatorname{Var}(f,P)=\lim_{N\rightarrow\infty}\operatorname{Var}\left(\revision{N^{-1/2}}\sum_{n=1}^Nf(X_n)\right)
$$
if the right-hand side exists. 
The existence of the right-hand side limit is a kernel-specific problem and not addressed here. 
Let $\lambda \in [0,1)$. 
As in \citet{MR3551794}, to avoid a kernel-specific argument, we consider a pseudo asymptotic variance 
$$
\operatorname{Var}_\lambda(f,P)=\|f_0\|^2+2\sum_{n=1}^\infty\lambda^n\langle f_0,P^n f_0\rangle,
$$
where $f_0=f-\Pi(f)$, which always exists. 
Under some conditions, $\lim_{\lambda\uparrow 1-}\operatorname{Var}_\lambda(f,P)=\operatorname{Var}(f,P)$. 
We can also define $\operatorname{Var}_\lambda(f,P_G)$ for $\Pi$-square integrable function $f$ on $E$ by considering 
$f((x,i))=f(x)$. 

\begin{proposition}[Theorem 3.17 of \citet{andrieu2019peskuntierney}]
Suppose that $f$ is $\Pi$-square integrable. Then for $\lambda\in [0,1)$, 
$\operatorname{Var}_\lambda(f,P_G)\le \operatorname{Var}_\lambda(f,P)$. 
\end{proposition}

By taking $\lambda\uparrow 1$, we can expect that the non-reversible kernel $P_G$ is better than $P$ in the sense of smaller asymptotic variance.



\subsection{Step-by-step instruction for creating a $\Delta$-guided Metropolis--Haar kernel}
\label{subsec:step-by-step}

Here is a set of necessary conditions to build a Haar-mixture kernel $Q_*$ and a Metropolis--Haar kernel $(Q_*,\Pi)$. 

\begin{enumerate}
    \item $G=(G,\times)$ is a locally compact topological group equipped with the Borel $\sigma$-algebra and the right Haar measure $\nu$.   
    \item State space $E$ is a left $G$-set. 
    \item $\mu$ is a $\sigma$-finite measure and $Q$ is $\mu$-reversible Markov kernel on $(E,\mathcal{E})$.
    \item There exists a Markov kernel $K(x,\dif g)$ as in (\ref{eq:markov-kernel-k}). 
\end{enumerate}
Then we can construct a Haar--mixture kernel $Q_*$ as in (\ref{eq:haar_mixture}). 
Here is an additional set of necessary conditions to build a $\Delta$-guided Metropolis--Haar kernel. 

\begin{enumerate}
    \item $G=(G,\le)$ is a\revision{n} ordered group, and $G=(G,\times)$ is a unimodular locally compact topological group. 
    \item $\Delta$ is a $G$-statistics. 
    \item $\Delta$ is sufficient for $Q$. 
\end{enumerate}

\subsection{Examples of $\Delta$-guided Metropolis--Haar kernels}
\label{subsec:example}

Here we present some of the $\Delta$-guided Metropolis--Haar kernels. 

\begin{example}[Guided Metropolis autoregressive mixture kernel]
\label{ex:guided-pcn}
\revision{
The Metropolis kernel of $(Q,\Pi)$ with the proposal kernel $Q$ defined in Example \ref{ex:ar-intro} is called the preconditioned Crank--Nicolson kernel.  This kernel was studied in \citet{MR1723510, MR2444507, MR3135540}. The Metropolis--Haar kernel with the Haar-mixture kernel $Q_*$ in Example \ref{ex:ar-mid} is called the mixed preconditioned Crank--Nicolson kernel. This kernel was developed in \citet{MR3668488,MR3788187}. 
The $\Delta$-guided Metropolis--Haar kernel of $(Q_*,\Pi)$ with $E=\mathbb{R}^d$ and $G=\mathbb{R}_+$, called the $\Delta$-guided mixed preconditioned Crank--Nicolson kernel, can be constructed as in Definition \ref{def:guided}. In this case, for a constant $x_0\in\mathbb{R}^d$ and a symmetric positive definite matrix $M$, 
$\Delta x=(x-x_0)^{\top}M^{-1}(x-x_0)$, $K(x,\dif g)=\mathcal{G}(d/2, \Delta x/2)$ and $Q_g(x, \dif y)=\mathcal{N}_d(x_0+(1-\rho)^{1/2}(x-x_0),g^{-1}\rho M)$ and $\mu_*(\dif x)\propto (\Delta x)^{-d/2}\dif x$. We can perform the $\Delta$-guided Metropolis--Haar kernel as in Algorithm \ref{alg:delta}. 
}
\end{example}

\begin{example}[Guided Metropolis Multivariate Beta-Gamma mixture kernel]
\label{ex:beta-gamma-guide}
\revision{
The Metropolis kernel of $(Q,\Pi)$ and the Metropolis--Haar kernel of $(Q_*,\Pi)$  in Example \ref{ex:product} can be defined naturally, and the former kernel was studied in \cite{MR4016139}. 
 The $\Delta'$-guided Metropolis--Haar kernel with $\Delta'(x)=x_1\times\cdots \times x_d$ is constructed by 
$K$, $Q_g$ and $\mu_*$ as in Example \ref{ex:product}. 
In this case, $E=G=\mathbb{R}^d_+$.
}
\end{example}

\begin{example}[Guided Metropolis Multivariate Chi-squared mixture kernel]
\label{ex:chisq-quide}
\revision{
The Metropolis kernel of $(Q,\Pi)$ and that of $(Q_*,\Pi)$ in Example \ref{ex:multi-chi-squared} can be defined naturally. The $\Delta$-guided kernel with $\Delta x=x_1+\cdots +x_d$ is constructed by $K$, $Q_g$ and $\mu_*$ as in Example \ref{ex:multi-chi-squared}. In this case, $E=\mathbb{R}^d_+$ and $G=\mathbb{R}_+$. 
}
\end{example}



\section{Simulation}
\label{sec:simulation}
\subsection{\revision{$\Delta$-guided Metropolis--Haar kernel on $\mathbb{R}^d$}}
\label{sec:multiplicative}

\revision{
In this simulation, we consider the autoregressive based kernel considered in Example \ref{ex:guided-pcn}. More precisely, we study the preconditioned Crank--Nicolson kernel, the mixed-preconditioned Crank--Nicolson kernel and the $\Delta$-guided mixed preconditioned Crank--Nicolson kernel. The random-walk Metropolis kernel is also compared for a reference. 
All these methods are gradient-free, information blind methods in the sense that the proposal kernel does not use the derivative of $\log\pi(x)$. 
Although this may sound daunting, sometimes a simple structure leads to robustness and efficiency as described through simulation experiments. Moreover, parameter tuning for these Markov kernels based on a reversible proposal kernel is relatively straightforward. We can learn the parameters of the reference measures $\mu$ or $\mu_*$ using the standard technique of treating the MCMC outputs as if they were from identically and independent observations of $\mu$ or $\mu^*$, even though $\mu^*$ is generally improper distribution. Since parameter tuning is not our main focus, we do not elaborate on this point in this paper. 
}

We also compare these methods with gradient based, informed algorithms. 
The Metropolis-adjusted Langevin algorithm \citep{doi:10.1063/1.436415,RT2} and the Hamiltonian Monte Carlo algorithm \citep{Duane1987216,MR2858447} are popular gradient based algorithms. Furthermore, we consider the methods that use both gradient-based and autoregressive kernel based ideas. This class includes such as the infinite dimensional Metropolis-adjusted Langevin algorithm \citep{MR2444507, MR3135540}, a marginal sampler proposed in \citep{MR3849342}, which we will refer to the marginal gradient-based sampling, and the infinite dimensional Hamiltonian Monte Carlo \citep{MR2858447,Ottobre_2016,BESKOS2017327}. 

\revision{We performed all experiments using a desktop computer with 6 cores Intel i7-5930K (3.50GHz) CPU. All algorithms other than the Hamilton Monte Carlo algorithm were coded in R version 3.6.3 \citep{rmanual} using the \textit{RcppArmadillo} package version 0.9.850.1.0 \citep{RcppArmadillo}. The results for the Hamilton Monte Carlo algorithm were obtained using \textit{rstan} version 2.19.3 \citep{rstan}. For a fair comparison, we use a single core and chain for \textit{rstan}. The code for all experiments is available in the online repository at the link https://github.com/Xiaolin-Song/Non-reversible-guided-Metropolis-kernel/.
}


\subsubsection{Discrete observation of stochastic diffusion process}
\label{5.1.1}
First we consider a problem which is difficult to apply gradient based Markov chain Monte Carlo methodologies due to high cost of derivative calculation. 
Let $\alpha\in\mathbb{R}^k$. 
Suppose that $(X_t)_{t\in [0,T]}$ is a solution process of a stochastic differential equation
$$
\dif X_t=a(X_t,\alpha)\dif t+b(X_t)\dif W_t; X_0=x_0
$$
where $(W_t)_{t\in [0,T]}$ is the $d$-dimensional standard Wiener process and $a:\mathbb{R}^d\times\mathbb{R}^k\rightarrow \mathbb{R}^d$ and 
$b:\mathbb{R}^d\rightarrow \mathbb{R}^{d\times d}$ are the drift and diffusion coefficient respectively. 
We only observe $X_0, X_h, X_{2h},\ldots, X_{Nh}$ where $N\in\mathbb{N}$ and $h=T/N$. 

We consider a Bayesian inference based on the local Gaussian approximate likelihood since explicit form of the probability density function is not available in general. 
The local Gaussian approximation approach, including simple least square estimate approach, has been studied in such as \cite{Prakasa_Rao_1983,MR999016,Florens_zmirou_1989,YOSHIDA1992220}. See also \cite{doi:10.1111/j.1467-9868.2006.00552.x,Beskos_2009} for non-local Gaussian approach based on unbiased estimate of the likelihood. 

We consider a Bayesian inference for $\alpha\in \mathbb{R}^{50}$ using local Gaussian approximated likelihood. 
We set the diffusion coefficient to be $b\equiv 1$, and the drift coefficient to be
$$
a(x,\alpha)=\frac{1}{2}\nabla\log\pi(x-\alpha);
$$
where \revision{$\pi(x)\propto 1/(1+x^{\top}\Sigma^{-1}x/20)^{35}$ where $\pi(x)$ here is the probability density function with respect to the Lebesgue measure}. See \cite{MR2038227}. \revision{Here $\Sigma$ is} generated from a Wishart distribution with $50$-degrees of freedom and the identity matrix as the scale matrix. The terminal time is $T=10$ and the number of observation is $N=10^3$. 
The prior distribution is a normal distribution $\mathcal{N}_{50}(0,10~I_{50})$. 

\begin{table*}[ht!]
\centering
	\caption{Markov kernels in Section \ref{sec:multiplicative}. The first four algorithms are gradient-free, information blind algorithms. The last five algorithms are gradient based, informed algorithms. }
	\label{alg}
	\ra{1}
	\begin{tabular}{ll}
		\textsc{rwm}& Random-walk Metropolis\\
		\textsc{pcn}& Preconditioned Crank--Nicolson\\
		\textsc{mpcn}& Mixed preconditioned Crank--Nicolson\\
		\textsc{gmpcn}& $\Delta$-guided mixed preconditioned Crank--Nicolson\\
		\hline
		\textsc{mala}& Metropolis-adjusted Langevin \\
		$\infty$-\textsc{mala}& Infinite dimensional Metropolis-adjusted Langevin\\
		\textsc{mgrad}& Marginal gradient-based sampling\\
		\textsc{hmc}& Hamiltonian Monte Carlo via \textit{rstan}\\
		$\infty$-\textsc{hmc}& Infinite dimensional Hamiltonian Monte Carlo
	\end{tabular}
\end{table*}


\begin{figure*}[ht]
	\centering
	\includegraphics[width=0.95\textwidth]{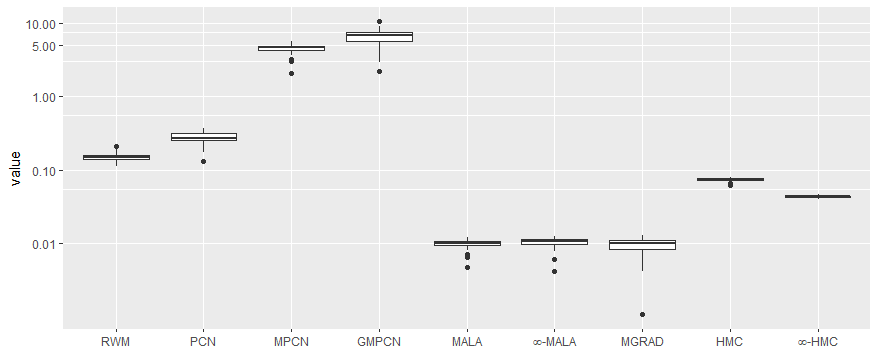}
	\caption{Effective sample sizes of log-likelihood per second of the stochastic diffusion process in Section \ref{5.1.1} for the nine Markov kernels listed in Table \ref{alg}. The $y$-axis is on a logarithmic scale. }
	\label{figb}
\end{figure*}

\begin{figure*}[ht]
	\centering
	\includegraphics[width=0.9\textwidth]{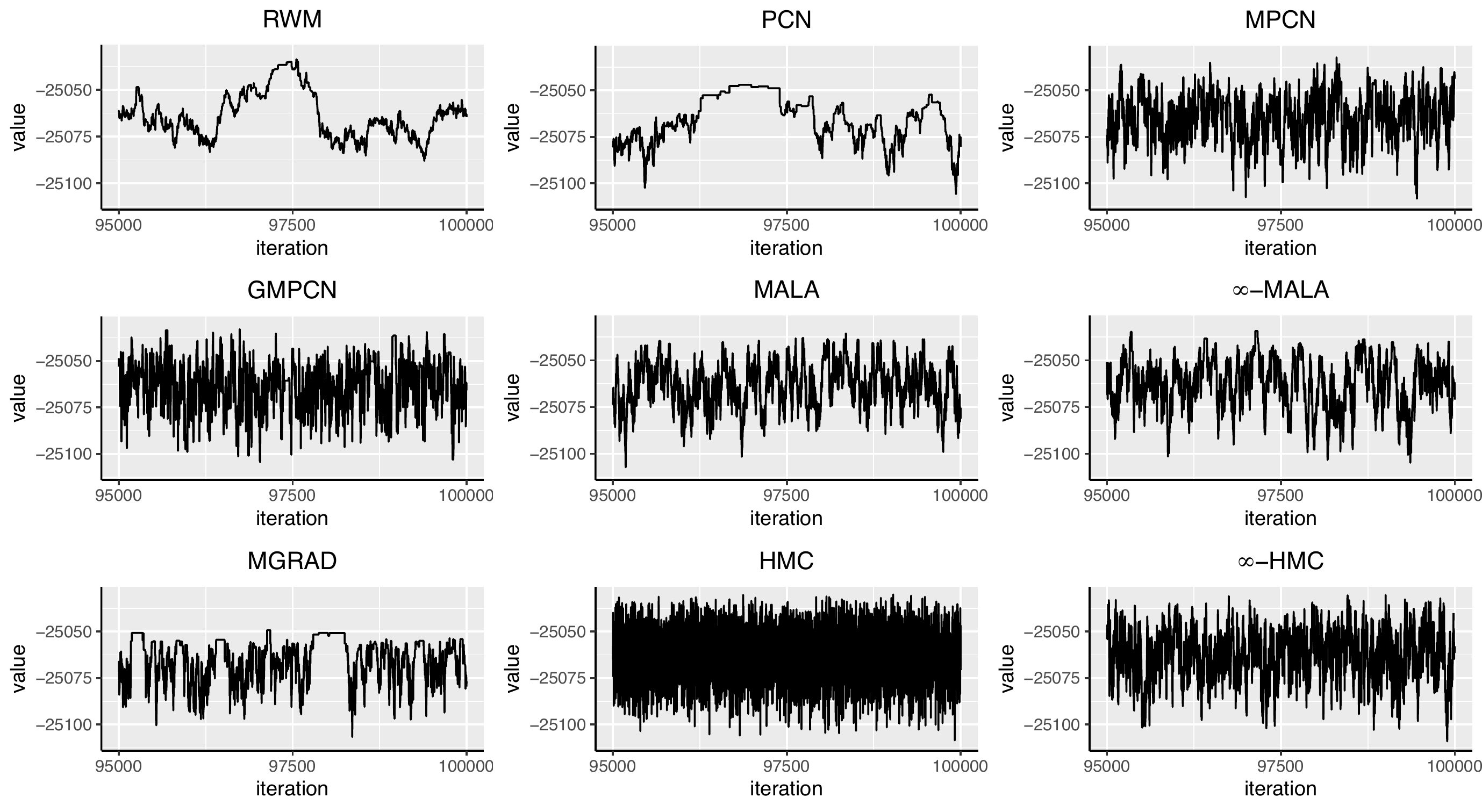}
	\caption{Trace plots of log-likelihood   of the stochastic diffusion process in Section \ref{5.1.1} for the nine Markov kernels listed in Table \ref{alg}.}
	\label{figp}
\end{figure*}

The Markov kernels used in this simulation is listed in Table \ref{alg}. 
The first four kernels in the table are gradient-free, information blind kernels. The last five kernels are gradient based, informed kernels. 
All kernels other than the 1st, 5th and 8th algorithms in Table \ref{alg} use the prior distribution as the reference distribution. 
\revision{Reference measure here means that the proposal kernel itself is reversible with respect to the measure, or the proposal kernel approximates another Markov kernel that is reversible with respect to the measure. }

We apply the Markov chain Monte Carlo algorithms by a 2-step procedure. In the first stage, we run the random-walk Metropolis algorithm as a burn-in stage. For Gaussian reference kernels, $x_0$ is estimated by the empirical mean in the burn-in stage. After the burn-in, we run each algorithm. The result was presented in Table \ref{figb} and Figure \ref{figp}. \revision{In this example, the covariance matrix is not preconditioned; we use the prior's covariance matrix instead.}

\revision{The acceptance rates for the first two algorithms in Table \ref{alg} were set at $25\%$. For the 3rd and 4th algorithms, acceptance rates were set to $30\%$ to $50\%$.
As suggested by \cite{roberts1998} and \cite{MR3849342}, the 5th through 7th algorithms, the acceptance probabilities were set to approximately $60\%$. The 8th algorithm was tuned in two steps. First, we set the number of leapfrog steps to $1$ and tune the leapfrog step size so that the acceptance rate is between $60\%$ and $80\%$ according to \cite{Beskos_2013}. Then we increase the number of leapfrog steps until the time-noramlised effective sample size decreases. The tuning parameters of the Hamiltonian Monte Carlo algorithm were controlled using \textit{rstan} package. As a quantitative measure of efficiency, we used the effective sample size of log-likelihood per second. It was estimated using the package \textit{coda} in R \citep{coda}.}

\revision{
The effective log-likelihood sample sizes per second are shown in Figure \ref{figb}. The box plot is constructed by fifty independent simulations for each algorithm. 
The 5th to 7th algorithms, which are Langevin diffusion based algorithms, show the worst performance. 
Due to the high cost of derivative evaluation, the Hamiltonian Monte Carlo and the infinite dimensional Hamiltonian Monte Carlo are still worse than the random-walk Metropolis kernel. The random-walk Metropolis kernel and the preconditioned Crank--Nicolson kernel are better than gradient-based kernels, 
but the mixed preconditioned Crank--Nicolson kernel is much better. The $\Delta$-guided version is even better than the non-$\Delta$-guided version thanks to the non-reversible property. 
A trace plot is also shown in the Figure \ref{fig14}, it illustrates that the Hamilton Monte Carlo method has a good performance per iteration, but the cost is high compared to other algorithms. 
}



\begin{figure*}[hbt!]
	\centering
	\includegraphics[width=0.9\textwidth]{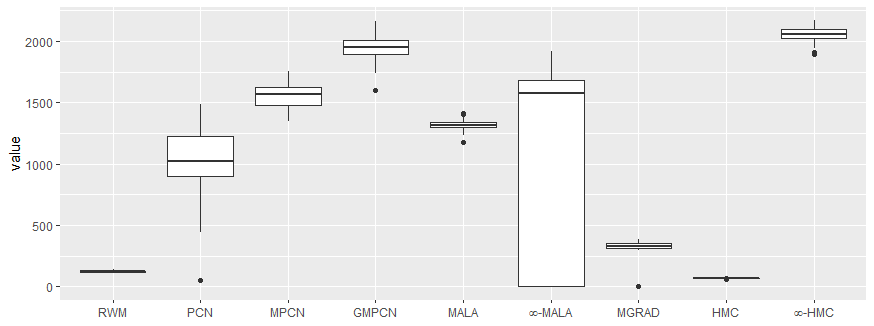}
	\caption{Effective sample sizes of log-likelihood per second in logistic regression example in Section \ref{logistic} for the nine Markov kernels listed in Table \ref{alg}}
	\label{fig13}
\end{figure*}

\begin{figure*}[hbt!]
	\centering
	\includegraphics[width=0.85\textwidth]{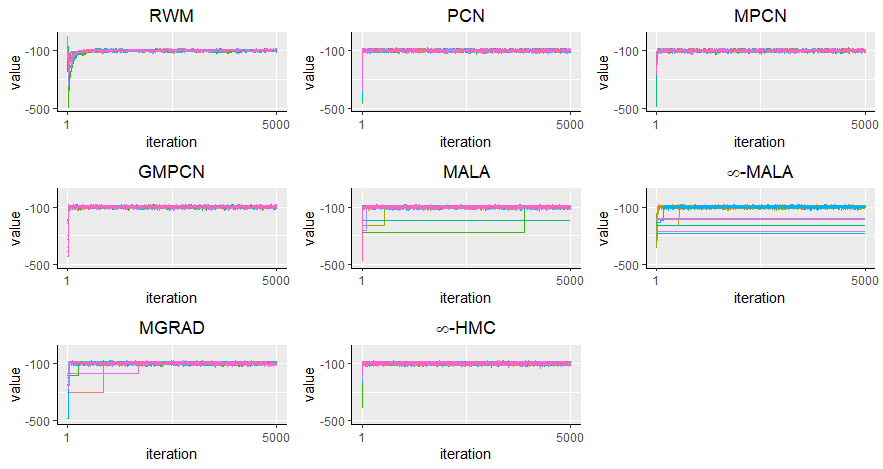}
	\caption{Sampling paths of the logistic regression example illustrated in Section \ref{logistic}. The Hamilton Monte Carlo algorithm is excluded from this simulation because the initial values of the algorithm are automatically selected in the \textit{rstan} package. }
	\label{fig14}
\end{figure*}

\subsubsection{Logistic regression}
\label{logistic}
Next we apply them to a logistic regression model with the \revision{Sonar data set} from
the University of California, Irvine repository \citep{Dua:2019}. The data set contains 208 observations and 60 explanatory variables. The prior distribution is $\mathcal{N}(0,10^2)$ for each parameters. We use a relatively large variance of the normal distribution because we did not have enough prior information at this stage.

\revision{Estimation of the preconditioning matrix is necessary for this problem due to the existence of a strong correlation between the variables.  We performed $2.0\times 10^5$ iterations to estimate $\mu_0$ and estimate the preconditioning matrix $\Sigma_0$ using the empirical means}.  Then we  ran $10^5$ iterations for each algorithm, discarding the $2\times 10^4$ iterations as burn-in. Furthermore, we ran each experiment for 50 times using different seeds. 
We evaluate the effective sample size of log-likelihood per second, and present the results of all the algorithms by boxplots (Figure \ref{fig13}).  
The algorithms based on the Lebesgue measure (1, 5, 8th algorithms in Table \ref{alg}) are relatively worse than other algorithms based on the Gaussian reference measure. The performances of the gradient-based algorithms are divergent, which might reflect the sensitivity of the gradient-based algorithms, which is well described in the \cite{Chopin_2017}. In particular, the infinite dimensional Hamiltonian Monte Carlo algorithm shows the better performance in this case, although it shows poor performance in the previous simulation. The $\Delta$-guided mixed preconditioned Crank--Nicolson kernel was slightly worse than infinite dimensional Hamiltonian Monte Carlo algorithm and better than all other algorithms. The Metropolis--Haar and $\Delta$-guided Metropolis--Haar kernels show good and robust results for the two simulation experiments. 

\revision{
We also investigate the sensitivity of the gradient-based algorithms for the same model as displayed in  Figure \ref{fig14}. In this example, $10$ initial values are randomly generated from a multivariate normal distribution for each algorithm. The number of iteration of each algorithm is $5 \times 10^3$. The paths of the gradient-based algorithms depend strongly on the initial values with the exception of the infinite dimensional Hamiltonian Monte Carlo algorithm. 
}

\subsubsection{\revision{Sensitivity} of the choice of $x_0$}
\label{5.1.3}

To illustrate the importance of $x_0$, we additionally run a numerical experiment on a $50$-dimensional multivariate central $t$-distribution with degrees of freedom $\nu=3$ and identity covariance matrix (\citealp[1p]{MR2038227}). The first element of $x_0$ is $\xi\ge 0$ and all the other elements are set to be zero. When $\xi$ is large, then the direction is less important for increasing or decreasing the likelihood. We run the algorithms on the target distribution for $10^5$ iterations. The experiment showed that the benefit of non-reversibility diminishes as the importance of the direction shrinks (Table \ref{tab2}). 
\begin{table*}[ht!]
    \centering
	\caption{Effective sample sizes of log-likelihood per second target on a 50-dimensional student distribution in Section \ref{5.1.3}}
    \label{tab2}
	\ra{1}
	\begin{tabular}{@{}cccccccccrrr@{}}
		
		& $\xi=0 $ & $\xi=10^{-3}$ & $\xi=10^{-2}$ & $\xi=10^{-1}$ 
		 & $\xi=1$ & $\xi=10$   \\
		\textsc{mpcn} &378.19&96.23&94.74&93.52&95.33&46.31 \\
		\textsc{gmpcn} &4245.43&116.29&114.78&115.2&117.20&40.20
	\end{tabular}

\end{table*}

\subsection{\revision{$\Delta$-guided Metropolis--Haar  kernels on $\mathbb{R}_+^d$}}
\label{sec:additive}

\revision{
Next, we consider the Beta-Gamma based kernels considered in Example \ref{ex:beta-gamma-guide} and the Chi-squared based kernels considered in example \ref{ex:chisq-quide} with $L=1$. 
Thus, we consider a total of six Markov kernels. These are the Metropolis kernel, the Metropolis--Haar kernel, and the $\Delta$-guided Metropolis--Haar kernel for each of the Beta-Gamma based and Chi-squared based kernels. }

\revision{
Our goal is not to compare the Beta-Gamma based kernels and the Chi-squared based kernels, but to compare the guided kernels and the non-guided kernels. In this simulation, we illustrate the difference in behaviour between the guided Metropolis kernel and other kernels by plotting trajectories in two dimensions. 
}

We consider a Poisson hierarchical model of the form
$$
x_{m,n}|\theta_m \sim \mathrm{Poisson}(\theta_m)\quad n=1,\ldots, N \\
$$
$$
\theta_m \sim \mathcal{G}(\alpha,\beta)\quad m=1,\ldots, M \\
$$
$$
\alpha \sim \mathcal{G}(1/20,1/20), \quad \beta \sim \mathcal{G}(1/20,1/20),
$$ where $x=\{x_{m,n}:m=1,\dots,M, n=1,\dots N\}$ is the observation. In our simulations we set $M=25$ and $N=5$. The number of unknown parameters is $M+2=27$ in this case. The parameter 
$\theta=(\theta_1,\ldots,\theta_M)$ has a closed form conditional distribution
$$
\theta_m|\alpha,\beta,x \sim \mathcal{G}\left(\sum_{n=1}^N x_{m,n}+\alpha,N+\beta\right) \quad m=1,\ldots,M.
$$
Therefore we can use the Gibbs sampler for generating the parameter $\theta$. 
On the other hand, since the conditional distribution of $\alpha, \beta$ is complicated, we apply Monte Carlo algorithms mentioned above. 
We created two-dimensional trajectory plots to illustrate the difference in behavior between the Metropolis--Haar kernel and its $\Delta$-guided version. 
The tuning parameters are chosen so that the average acceptance probabilities are $30\%$--$40\%$ in $50,000$ iterations. 
Figure \ref{fig1} shows the trace plots of the last $300$ iterations for the kernels. One can clearly see the larger variation for the guided kernels. Thanks to the incident variables, the guided kernel maintains its direction when the proposed value is accepted. The property of maintaining direction has greatly contributed to the increase in variability. 

\begin{table*}[hbt!]
\centering
	\caption{Description of Markov kernels in Figure \ref{fig1} in Section \ref{sec:additive}.}
	\label{alg2}
	\ra{1}
	\begin{tabular}{ll}
		\textsc{mh}& Metropolis\\
		\textsc{mhh}& Metropolis with Haar-mixture kernel\\
		\textsc{gmh}& Guided Metropolis\\
	\end{tabular}
\end{table*}

\begin{figure*}[hbt!]
	\centering
	\includegraphics[width=0.75\textwidth]{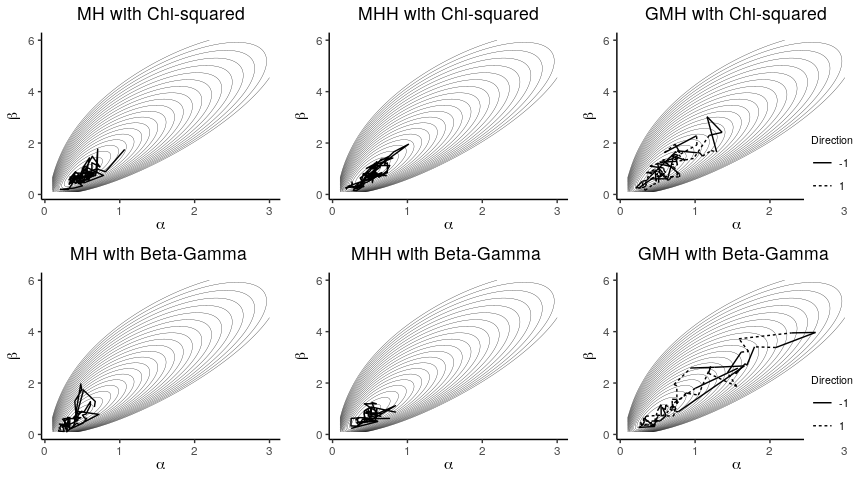}
	\caption{Trace plots of the Metropolis kernels in Section \ref{sec:additive}. The guided kernels (the right figures) are more variable compared to their non-guided counterparts where the solid line corresponds to the negative direction and the dashed line corresponds to the positive direction. }
	\label{fig1}
\end{figure*}



\section{Discussion}
\label{sec:discussion}

The theory and application of non-reversible Markov kernels have been under active development recently, but there still exists a gap between the two.
In order to close this gap, we have described how to construct a non-reversible Metropolis kernel on a general state space.   We believe that the method we propose can make non-reversible kernels more attractive. 

As a by-product, we have constructed the Metropolis--Haar kernel. The Haar-mixture kernel imposes a new state globally by using the random walk on a group, whereas other recent Markov chain Monte Carlo methods use local topological information derived from target densities. We believe that this sheds new light on the proposed gradient-free, global topological approach. A combination of the global and local (gradient-based) approaches is an area of further research.

In this paper, we have not discussed geometric ergodicity, although ergodicity is clear under appropriate regularity conditions. A popular approach for proving geometric ergodicity is based on the establishment of a Foster-Lyapunov-type drift condition, which requires kernel-specific arguments. On the other hand, our motivation is to build a general framework for the non-reversible Metropolis kernels. Therefore, we did not focus on geometric ergodicity. A more in-depth study should be carried out in that direction. See \cite{MR3668488} for geometric ergodicity of the mixed preconditioned Crank--Nicolson kernel. 

Finally, we would like to remark that 
the $\Delta$-guided Metropolis--Haar kernel is not limited to $\mathbb{R}^d$ or $\mathbb{R}_+^d$. It is possible to construct the kernel on the $p\times q$-matrix space and the symmetric $q\times q$ positive definite matrix space, where $p, q$ are any positive integers. 
$\Delta$-guided Metropolis--Haar kernels for other state spaces are future work.

\section*{Acknowledgements}
Kamatani is supported by JSPS KAKENHI Grant Number 20H04149 and JST CREST Grant Number JPMJCR14D7. Song is supported by the Ichikawa International Scholarship Foundation. We thank Sam Power for helpful comments.

\bibliographystyle{plainnat}

\end{document}